\documentclass[journal]{IEEEtran}

\usepackage[cmex10]{amsmath}
\usepackage{amssymb}
\usepackage{enumitem}

\usepackage{cases}
\usepackage{multirow}
\usepackage{empheq}

\usepackage{siunitx}
\usepackage[noadjust]{cite}
\usepackage{verbatim}

\usepackage{algorithm}
\usepackage{algpseudocode}

\ifCLASSINFOpdf
\usepackage[pdftex]{graphicx}
\else
\usepackage[dvips]{graphicx}
\fi

\usepackage[caption=false,font=footnotesize]{subfig}

\usepackage{color}
\usepackage{multirow}

\usepackage[hyphens]{url}
\usepackage{hyperref}
\usepackage{booktabs}
\newtheorem{theorem}{Theorem}
\newtheorem{lemma}{Lemma}

\newtheorem{proposition}{Proposition}

\newtheorem{example}{Example}

\title{Degree-Constrained Interval Optimization for Minimax Polynomial Approximation in Homomorphic Encryption}

\author{Jiheon Woo, Donggyun Ryu, and Yongjune Kim%
\thanks{J. Woo, D. Ryu, and Y. Kim are with the Department of Electrical Engineering, Pohang University of Science and Technology (POSTECH), Pohang 37673, South Korea (e-mail: \{jhwoo1997, dgryu, yongjune\}@postech.ac.kr).}%
}

\begin{document}
\sloppy

\maketitle

\begin{abstract}

Homomorphic encryption (HE) enables privacy-preserving inference under arithmetic constraints that restrict encrypted evaluation to additions and multiplications. 
As a result, non-polynomial activation functions must be replaced by polynomial approximations. 
Among polynomial approximation methods, minimax approximation, typically computed by the Remez algorithm, is a standard approach because it minimizes the maximum approximation error over a given design interval. 
For minimax polynomial design, the approximation interval is a critical hyperparameter: a wider interval improves robustness to large-magnitude inputs while increasing the minimax approximation error under a fixed degree budget. 
In this paper, we formulate this trade-off as a distribution-aware interval optimization problem, where the approximation interval is chosen to minimize the mean-squared error (MSE) with respect to the pre-activation distribution of interest.
To effectively control outside-interval inputs, we combine minimax polynomials with domain extension functions (DEFs) and their HE-realizable polynomial counterparts, domain extension polynomials (DEPs), which approximate a clipping operation outside the design interval and thereby suppress uncontrolled polynomial extrapolation.  
We first derive an analytically tractable DEF-based proxy objective that captures the trade-off between within-interval minimax approximation error and outside-interval clipping error. 
We then connect this idealized objective to HE-realizable DEP constructions through an implementation-error decomposition with an accompanying upper bound. 
Numerical experiments on representative non-polynomial activation functions show that the proposed interval optimization achieves significantly lower MSE than conventional minimax baselines, with particularly large gains for sigmoid, tanh, and GELU, while the minimizer of the analytical proxy closely matches that of the numerical ideal objective.
\end{abstract}

\begin{IEEEkeywords}
Homomorphic encryption, Privacy-preserving machine learning, Polynomial approximation, Interval optimization
\end{IEEEkeywords}

\section{Introduction}\label{sec:introduction}

The rapid advancement of artificial intelligence (AI) has accelerated the adoption of machine learning as a service (MLaaS), in which clients delegate inference tasks to remote servers.
However, this paradigm raises significant privacy concerns with respect to the handling of sensitive user data. 
To address these concerns, privacy-preserving machine learning (PPML) has emerged as a key research area, aiming to perform inference while keeping user data confidential~\cite{Al-Rubaie2019Privacy,Riazi2019deep,Juvekar2018Gazelle,Reagen2021Cheetah}.

Among various approaches to PPML, homomorphic encryption (HE) is regarded as one of the most promising solutions~\cite{Gentry2009fully,Cheon2017ckks}. 
HE allows a server to perform computations directly on encrypted data without decryption. 
This property naturally supports a non-interactive inference setting: once the user provides the encrypted input, the server can execute the entire inference pipeline and return the encrypted result~\cite{Gilad2016cryptonets,Dathathri2019chet}. 
Compared with interactive protocols such as multi-party computation (MPC)~\cite{Evans2018pragmatic}, this setting avoids repeated client-server communication and active client participation during inference.

A major practical challenge in HE-based PPML is that most HE schemes support only a limited set of operations: addition, multiplication, and rotation. 
Consequently, non-polynomial operations in neural networks, such as activation functions, cannot be directly evaluated under HE.
A common approach is therefore to replace such functions with polynomial approximations, enabling the inference pipeline to be represented as an arithmetic circuit.
A broad class of HE-oriented polynomial approximation methods has been developed to balance numerical accuracy and computational efficiency~\cite{Lee2021minimax,Lee2022privacy,Lee2023precise}.
Since the polynomial degree directly affects both approximation accuracy and homomorphic evaluation cost, designing low-degree yet accurate polynomial approximations remains a central challenge.

Minimax approximation is a standard tool for replacing non-polynomial functions with degree-constrained polynomials by controlling the worst-case error over a given design interval \cite{trefethen2019atap,Pachon2009Remez}. 
In principle, the minimax approximation error can be reduced by increasing the polynomial degree.
Under HE constraints, however, a higher degree generally leads to a larger homomorphic evaluation cost.
In particular, in the CKKS scheme \cite{Cheon2017ckks}, ciphertext--ciphertext multiplications are the primary source of noise growth, so feasibility is governed by the multiplication count and, more critically, the multiplicative depth. 
If the required depth exceeds the available noise budget, the evaluation may require costly bootstrapping \cite{Cheon2018bootstrapping}.
Accordingly, reducing the minimax error solely by raising the polynomial degree is often impractical in HE-based inference. 


Another fundamental limitation is that minimax approximation provides worst-case error control only within the design interval; outside this interval, the polynomial can deviate from the target function by an arbitrarily large margin.
Fig.~\ref{fig:remez_example} illustrates the sensitivity of polynomial approximations of activation functions.
A degree-\SI{14} minimax polynomial approximates ReLU accurately on the design interval $[-1,1]$.
However, when evaluated at an outside-interval input, the polynomial can produce a dramatically amplified error: at $x=1.5$, the approximation incurs an absolute error of $\approx 4.09\times 10^3$, whereas $\mathrm{ReLU}(1.5)=1.5$.
In other words, instead of producing an activation close to $1.5$, the polynomial evaluation may produce an error thousands of times larger than the target value.
In HE, intermediate plaintext values are not observable, so such amplified activation errors can propagate through subsequent layers and significantly degrade inference accuracy.


A common response is to enlarge the design interval, but this introduces a competing trade-off.
For a given polynomial degree, enlarging the interval improves coverage of large-magnitude inputs and reduces the risk of extrapolation beyond the design interval, at the cost of larger minimax error within the interval.
Indeed, in the large-interval setting studied in \cite{Cheon2022dep}, maintaining a fixed small worst-case error by minimax approximation can require polynomial degree $\Omega(t)$ on $[-t,t]$.
Thus, under a given degree budget, an overly conservative interval can substantially degrade approximation fidelity in the region of interest. 


Beyond simple interval enlargement, the behavior of the approximation outside the design interval can be controlled through the framework of domain-extension functions (DEFs) and domain-extension polynomials (DEPs) \cite{Cheon2022dep}.
Intuitively, a DEF provides an effective interval-control mechanism that preserves the desired function behavior on the design interval $[-t,t]$, while saturating the effective response beyond this interval so that large-magnitude inputs do not lead to unstable polynomial extrapolation.
A DEP is a polynomial realization of this mechanism, designed to be compatible with homomorphic polynomial evaluation. 


This extension-based approach, however, does not remove the need to optimize the approximation interval. 
Rather, it changes the failure mode from uncontrolled polynomial extrapolation to a controlled clipping error outside the design interval. 
Thus, the interval parameter $t$ still determines a fundamental trade-off: a smaller $t$ improves the minimax approximation within the interval $[-t,t]$ under a given degree budget, but increases the probability mass assigned to the clipped region; a larger $t$ reduces the clipping error at the expense of a larger minimax error within the interval.  
Therefore, even when DEFs or DEPs are used, the design interval should be selected by optimizing the total MSE.


Accordingly, we formulate an interval optimization problem in which the approximation interval $[-t,t]$ is chosen to minimize the mean-squared error (MSE) over the pre-activation distribution of interest. 
Our framework decomposes the MSE into a within-interval approximation term and an outside-interval term.
For the within-interval term, we relate the MSE to the minimax error through the equioscillation characterization, which yields an analytically tractable surrogate objective.
For the outside-interval term, DEFs and DEPs provide a controlled saturation mechanism that replaces unstable extrapolation with a bounded clipping error.
This leads to a two-level design framework: an ideal DEF-based proxy that provides an analytically tractable model of the trade-off between within-interval minimax error and clipping error, and a DEP-based realization that quantifies the additional gap induced by replacing the ideal hard-clipping DEF with an HE-realizable DEP composition.

The main contributions of this paper are summarized as follows:
\begin{itemize}
    \item To the best of our knowledge, we provide the first principled formulation of interval optimization for degree-constrained minimax polynomial approximation by minimizing the MSE over the pre-activation distribution. 
    \item We derive a DEF-based objective that captures the trade-off between within-interval minimax error and outside-interval clipping error in an analytically tractable form.
    \item We establish a connection between the ideal DEF analysis and HE-realizable DEP constructions through an implementation-error decomposition, and derive an upper bound on the DEP-induced MSE gap.
    \item We numerically compare the analytical DEF-based proxy, the numerical ideal DEF objective, and practical DEP realizations for representative activation functions.
\end{itemize}

The rest of this paper is organized as follows.
Section~\ref{sec:preliminaries} introduces preliminaries. 
Section~\ref{sec:intervalopt} presents the proposed interval optimization objective and practical proxy forms based on the ideal DEF.
Section~\ref{sec:dep_impl} introduces the HE-realizable DEP class for interval control and provides its error decomposition.
Section~\ref{sec:experiments} presents numerical experiments. 
Section~\ref{sec:conclusion} concludes. 

\begin{figure}[!t]
    \centering
    \includegraphics[width=\columnwidth]
        {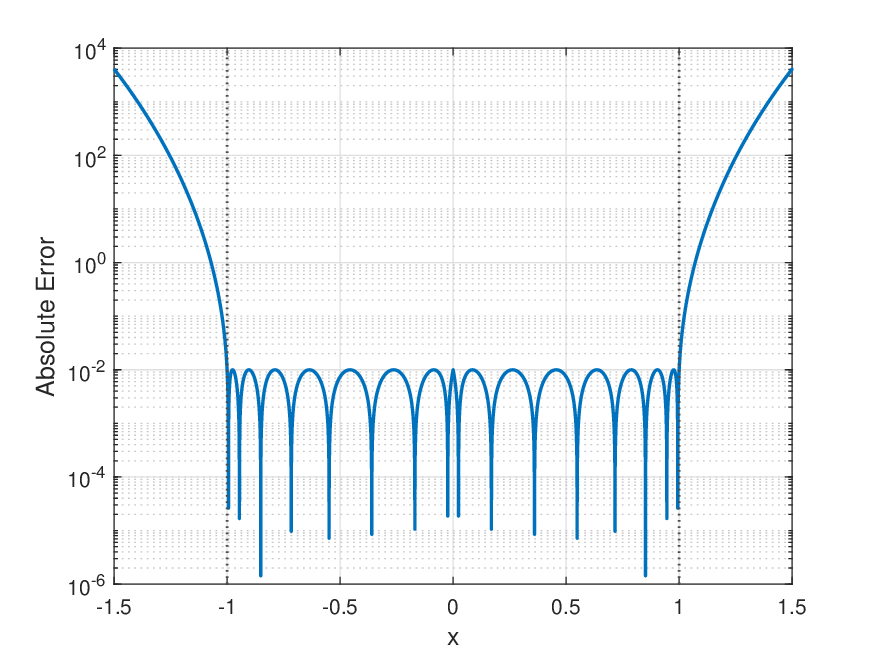}
    \caption{The absolute error of minimax polynomial approximation
    for $\mathrm{ReLU}(x)$ with degree $n=14$ on the approximation
    interval $[-1,1]$.}
    \label{fig:remez_example}
    \vspace{-4mm}
\end{figure}

\section{Preliminaries} \label{sec:preliminaries}

\subsection{Homomorphic Encryption}

HE is a cryptographic primitive that enables computations on encrypted data, yielding an encrypted output that decrypts to the corresponding plaintext result.
In PPML, approximate-arithmetic HE schemes for real-valued data, such as the CKKS scheme \cite{Cheon2017ckks}, are widely adopted.
The supported homomorphic operations include element-wise arithmetic operations such as addition and multiplication; in packed schemes, rotations are also supported.
Regarding multiplication, it is common to distinguish ciphertext--plaintext multiplication (PMult), which involves a plaintext constant, from ciphertext--ciphertext multiplication (CMult), which is typically the computational bottleneck and a primary source of noise growth. 
Accordingly, computational cost and feasibility are governed largely by the multiplicative depth and the noise growth induced by repeated ciphertext multiplications. When the multiplicative depth approaches the available noise budget, costly bootstrapping operations are required to refresh the ciphertext \cite{Cheon2018bootstrapping}.

Such operational constraints have immediate consequences for model design under HE.
Common ML components such as ReLU, sigmoid, tanh, normalization, and comparisons involve non-polynomial or non-arithmetic operations and therefore cannot be directly evaluated using native HE operations.
A standard approach is to replace each non-polynomial primitive with a polynomial approximation, so that the entire inference pipeline becomes an arithmetic circuit. 
In this setting, approximation design plays a critical role in determining the behavior of encrypted inference.
Because intermediate plaintext values are not observable during encrypted execution, the approximation interval must be chosen to control out-of-interval behavior while keeping the polynomial degree and multiplicative depth feasible under HE.

\subsection{Minimax Approximation}\label{subsec:minimax}

A widely used baseline for polynomial approximation under HE is the \emph{minimax} approximation~\cite{trefethen2019atap}, which directly controls the worst-case error over a prescribed design interval. 
Given a continuous function $f$ on $[a,b]$ and the space $\Pi_n$ of polynomials of degree at most $n$, the minimax problem is
\begin{equation}
p_n^* :=\arg\min_{p \in \Pi_n} \max_{x \in [a,b]} \bigl| f(x)-p(x) \bigr| . 
\end{equation}
The corresponding \emph{uniform} (worst-case) approximation error is defined as
\begin{equation}
E_n \;:=\; \max_{x \in [a,b]} \bigl| f(x) - p_n^*(x) \bigr|. 
\end{equation}
The Remez algorithm is an iterative method for computing $p_n^*$ by exploiting the classical equioscillation characterization of minimax solutions~\cite{trefethen2019atap, Pachon2009Remez}. 
Except in the zero-error case, the optimal error function $f(x)-p_n^*(x)$ alternates in sign and attains equal magnitude at $n+2$ extremal points.
In HE-oriented design, minimax polynomials have been widely used to approximate non-polynomial primitives such as sign, ReLU, and max-pooling functions because they provide uniform error control over a prescribed interval~\cite{Lee2021minimax,Lee2022privacy,Lee2023precise}.
This differs from least-squares criteria, which minimize an integrated or average error and therefore do not directly constrain the maximum deviation~\cite{trefethen2019atap}.

\subsection{Domain Extension Functions and Polynomials}\label{subsec:def}

A key practical challenge in HE-oriented approximation is that polynomial approximations provide controlled accuracy only on a prescribed design interval. 
Outside this interval, the polynomial values can deviate from the target function values by an arbitrarily large margin, since the polynomial is no longer constrained by the minimax criterion outside its design range; this is particularly problematic in encrypted execution, where intermediate plaintext values cannot be monitored directly.

To address this challenge, the domain extension methodology of~\cite{Cheon2022dep} enlarges the input interval by introducing a DEF that maps a wider input interval to a narrower one, while preserving the behavior of the original approximation in a smaller region.
The DEF is defined as
\begin{equation}
D(x) \;=\; \frac{1}{2}\Bigl(|x+1| - |x-1|\Bigr).
\end{equation}
While $D(x)$ is not a polynomial and thus cannot be evaluated under HE, the idea of domain extension can be approximated by polynomials, leading to the introduction of DEPs.

A polynomial $d(\cdot)$ is a DEP in the class $\mathcal{D}(\delta, r, R_1, R_2)$ if it satisfies~\cite{Cheon2022dep}:
\begin{align}
|x - d(x)| &\le \delta |x|^3, && \forall x \in [-r, r], \label{eq:dep1} \\
0 \le d'(x) &\le 1, && \forall x \in [-r, r], \label{eq:dep2} \\
d(r) < d(x) &< R_1, && \forall x \in (r, R_2],\label{eq:dep3} \\
-R_1 < d(x) &< d(-r), && \forall x \in [-R_2, -r).\label{eq:dep4}
\end{align}
Here, $\delta$ controls how closely the DEP follows the identity map on $[-r,r]$.
The parameter $r$ specifies the radius of this preservation region. 
The parameters $R_1$ and $R_2$ specify the controlled output interval and the extended input interval, respectively. 
In particular, the conditions in \eqref{eq:dep3} and \eqref{eq:dep4} ensure that inputs in the extended regions $(r,R_2]$ and $[-R_2,-r)$ are mapped into the bounded interval $(-R_1,R_1)$, thereby avoiding unstable extrapolation for large-magnitude inputs.
The extension factor $L$ is defined as
\begin{equation}
L := \frac{R_2}{R_1}.
\end{equation}
This extension factor also provides a scaling rule for composing multiple DEP stages. 
Let $B(\cdot)$ denote a fixed base DEP satisfying  $B(\cdot)\in\mathcal{D}(\delta,r,R,LR)$ for some $L>1$. 
For each stage $i$, define the scaled DEP
\begin{equation}
B_i(x):=L^i B\left(\frac{x}{L^i}\right),
\qquad i=0,1,\ldots,m-1 .
\end{equation}
The $m$-stage DEP composition is then given by
\begin{equation}
g^{(m)}(x)
:=
B_0\circ B_1\circ\cdots\circ B_{m-1}(x).
\end{equation}
Through this scaling-and-composition structure, a base DEP that extends the admissible interval by a factor of $L$ can be repeatedly applied to cover an exponentially larger input interval as $m$ increases, while maintaining near-identity behavior in the preservation region $[-r,r]$.
Thus, DEPs provide an HE-compatible mechanism for controlling out-of-interval behavior.
However, the DEP framework itself does not determine the target design interval; the design interval should be selected.

\section{Interval Optimization via DEF}\label{sec:intervalopt}

In this section, we develop an interval optimization framework in an idealized setting based on the DEF.
Let $f:\mathbb{R}\to\mathbb{R}$ denote the target non-polynomial function, and let $X$ denote the pre-activation random variable.
For each design radius $t>0$, let $p_{n,t}\in\Pi_n$ denote the \emph{degree-$n$ minimax polynomial on $[-t,t]$}, defined as
\begin{equation}\label{eq:pnt-def}
p_{n,t} := \arg\min_{p\in\Pi_n}\max_{x\in[-t,t]}\bigl|f(x)-p(x)\bigr|.
\end{equation}
The design interval is then determined by optimizing $t$ to minimize the induced MSE objective.

Throughout the paper, we distinguish three progressively more practical objectives.
Since the polynomial degree $n$ and scaling factor $L$ are fixed throughout the interval-selection problem, we suppress the dependence on $n$ and $L$ in the objective notation unless otherwise needed.
First, $J_{\mathrm{DEF}}(t)$ denotes the ideal clipping objective based on the non-polynomial hard-clipping map.
Second, $\widehat{J}_{\mathrm{DEF}}(t)$ denotes an analytically tractable surrogate of $J_{\mathrm{DEF}}(t)$, introduced to characterize the interval-selection mechanism.
Third, $J_{\mathrm{DEP}}(t)$ denotes the practical objective induced by DEP compositions that are realizable under HE.
The theoretical analysis focuses on $\widehat{J}_{\mathrm{DEF}}(t)$, while the numerical validation compares it with $J_{\mathrm{DEF}}(t)$ and the HE-realizable objective $J_{\mathrm{DEP}}(t)$.

Let $D(\cdot)$ denote the normalized hard-clipping DEF defined in Section~\ref{subsec:def}. 
For each design radius $t>0$, we use its scaled version

\begin{equation}
D_t(x)
:=
tD\left(\frac{x}{t}\right)
= \begin{cases} -t, & x < -t, \\ x, & -t \le x \le t, \\ t, & x > t. \end{cases}
\end{equation}
Using this map, we define the clipped polynomial evaluator as
\begin{equation}\label{eq:q_def_core}
q_{n,t,\mathrm{DEF}}(x)\;:=\;p_{n,t}\!\bigl(D_t(x)\bigr).
\end{equation}
We then define the DEF-based MSE objective as
\begin{equation}\label{eq:J_def_def}
J_{\mathrm{DEF}}(t)\;:=\;\mathbb{E}\!\left[\bigl(f(X)-q_{n,t,\mathrm{DEF}}(X)\bigr)^2\right].
\end{equation}
The interval selection problem is to minimize $J_{\mathrm{DEF}}(t)$ over $t>0$.
Using the piecewise form of $D_t$, the objective decomposes into within-interval and outside-interval terms:
\begin{align}
J_{\mathrm{DEF}}(t)
&=
\int_{-t}^{t}\bigl(f(x)-p_{n,t}(x)\bigr)^2h_X(x)\,dx \nonumber\\
&+\int_{t}^{\infty}\bigl(f(x)-p_{n,t}(t)\bigr)^2h_X(x)\,dx \nonumber\\
&+\int_{-\infty}^{-t}\bigl(f(x)-p_{n,t}(-t)\bigr)^2h_X(x)\,dx,
\end{align}
where $h_X(x)$ denotes the probability density function (PDF) of $X$.
We denote the first integral by $J_{\mathrm{DEF},\mathrm{in}}(t)$ and the sum of the last two integrals by $J_{\mathrm{DEF},\mathrm{out}}(t)$. 
To obtain a tractable outside-interval surrogate, we substitute the target endpoint values $f(t)$ and $f(-t)$ for the polynomial endpoint values $p_{n,t}(t)$ and $p_{n,t}(-t)$, respectively.
This endpoint-value substitution yields the outside-interval surrogate
\begin{align}\label{eq:J_out_anchored}
\widehat{J}_{\mathrm{DEF},\mathrm{out}}(t)
&:=
\int_t^\infty (f(x)-f(t))^2h_X(x)\,dx
\nonumber\\&+
\int_{-\infty}^{-t} (f(x)-f(-t))^2h_X(x)\,dx .   
\end{align}

\begin{figure}[!t]
    \centering
    \includegraphics[width=\columnwidth]
        {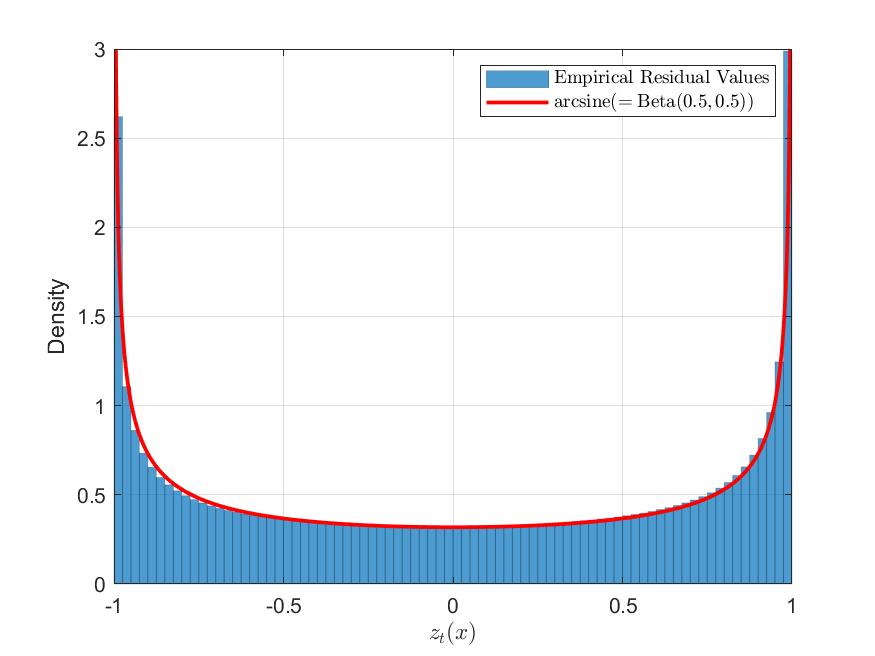}
    \caption{Comparison of the empirical distribution of the normalized residual values $z_t(x)=e_t(x)/E_n(t)$ with the arcsine law for the degree-$15$ minimax approximation of ReLU on the design interval $[-t, t]$ with $t=3$.}
    \label{fig:empirical_arcsine}
    \vspace{-4mm}
\end{figure}

For the within-interval term, we define the minimax error as
\begin{equation}\label{eq:minimax_definition}
E_n(t):=\|f-p_{n,t}\|_{\infty,[-t,t]} .
\end{equation}
Let $e_t(x):=f(x)-p_{n,t}(x)$ denote the corresponding residual, and let $z_t(x):=e_t(x)/E_n(t)\in[-1,1]$ denote its normalized version.
We model the residual values over the design interval as arcsine-distributed: 
\begin{equation}\label{eq:error-arcsine}
e_t(X)\mid \{|X|\le t\}\;\sim\;\operatorname{Arcsine}\!\left[-E_n(t),\,E_n(t)\right].
\end{equation}
As shown in Fig.~\ref{fig:empirical_arcsine}, the empirical distribution of the residual values closely follows the arcsine density $h_t(v)=1/\big(\pi\sqrt{E_n(t)^2-v^2}\big)$ for $v\in(-E_n(t),E_n(t))$.
Similar behavior is observed for the other activation functions and design radii in Section~\ref{sec:experiments}, supporting this modeling choice.

This model is motivated by the equioscillatory structure of minimax residuals. 
By the equioscillation characterization of minimax approximation~\cite{trefethen2019atap,Pachon2009Remez}, the residual $e_t$ is uniformly bounded by $E_n(t)$ on $[-t,t]$ and attains the values $+E_n(t)$ and $-E_n(t)$ with alternating signs at the $n+2$ alternation points.
The normalized residual $z_t(x)$ therefore attains the values $-1$ and $1$ with alternating signs at these $n+2$ points. 
The same alternation structure appears in the Chebyshev polynomial $T_{n+1}(x/t)=\cos\!\big((n+1)\arccos(x/t)\big)$, which attains alternating values of $-1$ and $1$ at its $n+2$ extremal points on $[-t,t]$.
Since this Chebyshev-type oscillation is exactly a cosine in the phase variable $(n+1)\arccos(x/t)$, and the values of a cosine with uniformly distributed phase follow the arcsine law, we approximate the residual-value distribution by the arcsine distribution in \eqref{eq:error-arcsine}.
Under this model, the conditional second moment is
\begin{align}\label{eq:arcsine_moment}
\mathbb{E}\!\left[e_t(X)^2\mid|X|\le t\right]
&=\int_{-E_n(t)}^{E_n(t)}\frac{v^2}{\pi\sqrt{E_n(t)^2-v^2}}\,dv \\
&=\frac{E_n(t)^2}{\pi}\int_{0}^{\pi}\cos^2\theta\,d\theta \label{eq:arcsine_moment1}\\
&=\frac{E_n(t)^2}{2},
\end{align}
where \eqref{eq:arcsine_moment1} follows from $v=E_n(t)\cos\theta$.
This yields the within-interval proxy term
\begin{equation}\label{eq:J_in_proxy}
\widehat{J}_{\mathrm{DEF},\mathrm{in}}(t):=\frac{E_n(t)^2}{2}\,\Pr(|X|\le t).
\end{equation}

Combining the within-interval term with the outside-interval term in \eqref{eq:J_out_anchored}, we define the proxy DEF objective as
\begin{align}\label{eq:J_def_hat}
\widehat{J}_{\mathrm{DEF}}(t)
&:=\widehat{J}_{\mathrm{DEF},\mathrm{in}}(t)+\widehat{J}_{\mathrm{DEF},\mathrm{out}}(t)\nonumber\\
&=\frac{E_n(t)^2}{2}\,\Pr(|X|\le t)+\widehat{J}_{\mathrm{DEF},\mathrm{out}}(t).
\end{align}
The only remaining quantity in the proxy objective \eqref{eq:J_def_hat} is the minimax error $E_n(t)$ as a function of the design radius $t$.
How this relation is obtained depends on the activation type: for positively homogeneous activations such as ReLU and LeakyReLU, the entire error--radius curve is fixed by a single numerically computed value $E_n(1)$ through an exact scaling argument (Section~\ref{subsec:def_homog}); for non-homogeneous activations such as sigmoid, tanh, and GELU, no such scaling is available, and $E_n(t)$ must instead be sampled over a range of radii and approximated by a fitted surrogate (Section~\ref{subsec:def_nonhomog}).

\subsection{Analytical Case: Homogeneous Functions}\label{subsec:def_homog}

We first consider a simple scaling property that will be used to obtain an analytical form of the proxy objective for positively homogeneous activations.
\begin{proposition}\label{proposition:homogeneous}
If $f(\cdot)$ is positively homogeneous of degree one, i.e.,
$f(\lambda x)=\lambda f(x)$ for all $\lambda>0$, then, for every
$t>0$, the minimax error on $[-t,t]$ satisfies $E_n(t)=tE_n(1)$.
\end{proposition}

\begin{IEEEproof}
Let $s^*_{n,1}\in\Pi_n$ be a minimizer on $[-1,1]$. 
Define $s_{n,t}(x):=t\,s^*_{n,1}(x/t)$ for $x\in[-t,t].$
Since $s^*_{n,1}\in\Pi_n$, we also have $s_{n,t}\in\Pi_n$. 
Writing $x=tu$, where $u\in[-1,1]$, and using the positive
homogeneity of $f$, we obtain
\begin{align}
E_n(t) &\le \|f-s_{n,t}\|_{\infty,[-t,t]} \nonumber \\
&=
\max_{|u|\le 1}
\left| f(tu)-t\,s^*_{n,1}(u) \right| \nonumber\\
&=
\max_{|u|\le 1}
\left| t f(u)-t\,s^*_{n,1}(u) \right| \nonumber\\
&=
t\|f-s^*_{n,1}\|_{\infty,[-1,1]}
=
tE_n(1).
\end{align}
Hence, $E_n(t)\le tE_n(1)$.

Conversely, let $s\in\Pi_n$ be arbitrary and define $\widetilde{s}_{n,1}(u):=\frac{1}{t}s(tu)$ for $u\in[-1,1]$.
Then $\widetilde{s}_{n,1}\in\Pi_n$. Again writing $x=tu$, we have
\begin{align}
\|f-s\|_{\infty,[-t,t]}
&=
\max_{|u|\le 1}
\left| f(tu)-s(tu) \right| \nonumber\\
&=
\max_{|u|\le 1}
\left| t f(u)-t\,\widetilde{s}_{n,1}(u) \right| \nonumber\\
&=
t\|f-\widetilde{s}_{n,1}\|_{\infty,[-1,1]}
\ge
tE_n(1).
\end{align}
Since this holds for every $s\in\Pi_n$, taking the infimum over
$s$ yields $E_n(t)\ge tE_n(1)$. Combining the two inequalities gives
$E_n(t)=tE_n(1)$.
\end{IEEEproof}

To parameterize the design radius relative to the input distribution, we define the normalized radius $\alpha:=t/\sigma$, where $\sigma$ denotes the standard deviation of the input distribution.
We next derive a characterization of the proxy objective under a Gaussian pre-activation model.
The Gaussian model is used here as an analytically tractable reference case, since the required tail probabilities and truncated moments admit closed-form expressions; for other input densities, the same DEF objective and proxy construction can be evaluated numerically.
\begin{theorem}
\label{thm:homogeneous_convexity}
Let $X\sim\mathcal{N}(0,\sigma^2)$ and let $f$ be positively homogeneous of degree one, with $B_f:=f(1)^2+f(-1)^2>0$ and $\varepsilon_n:=E_n(1)>0$.
Under the endpoint-value substitution, the proxy objective $\widehat{J}_{\mathrm{DEF}}(t)$ is strictly convex on $(0,\infty)$ and admits a unique global minimizer $t^*=\sigma\alpha^*$ on $(0,\infty)$, where the optimal normalized radius $\alpha^*$ is the unique positive solution of
\begin{equation}\label{eq:DEF_global_optimal}
\frac{\varepsilon_n^2}{2} \Bigl[ \alpha \bigl( 2\Phi(\alpha) - 1 \bigr) + \alpha^2\phi(\alpha) \Bigr] + B_f \bigl[ \alpha Q(\alpha) - \phi(\alpha) \bigr] = 0,\end{equation}
where $\phi(x)$ is the standard normal PDF, and $\Phi(\cdot)$ denotes the standard normal cumulative distribution function (CDF) with $Q(\alpha) := 1 - \Phi(\alpha)$.
\end{theorem}

\begin{IEEEproof}
By Proposition~\ref{proposition:homogeneous}, the within-interval proxy term becomes
\begin{align}\label{eq:J_DEF_inner}
\frac{E_n(t)^2}{2}\Pr(|X|\le \sigma\alpha)
&=
\frac{\sigma^2\alpha^2\varepsilon_n^2}{2}\bigl(2\Phi(\alpha)-1\bigr).
\end{align}
Under the endpoint-value substitution, the outside-interval term is
\begin{align}
\widehat J_{\mathrm{DEF,out}}(\sigma\alpha)
&=
f(1)^2\int_t^\infty (x-t)^2\phi_\sigma(x)\,dx \nonumber\\
&+
f(-1)^2\int_{-\infty}^{-t}(x+t)^2\phi_\sigma(x)\,dx \label{eq:hatJ_outside1}\\ 
&=\sigma^2 B_f\Bigl[(1+\alpha^2)Q(\alpha)-\alpha\phi(\alpha)\Bigr],\label{eq:hatJ_outside2}
\end{align}
where \eqref{eq:hatJ_outside1} follows from the positive homogeneity of $f$, $f(t)=tf(1)$ and $f(-t)=tf(-1)$ for $t>0$, and $\phi_\sigma(\cdot)$ denotes the PDF of $\mathcal{N}(0,\sigma^2)$.
Combining \eqref{eq:J_DEF_inner} and \eqref{eq:hatJ_outside2}, we obtain a closed-form expression for the proxy objective $\widehat{J}_{\mathrm{DEF}}(\sigma\alpha)$. 
Differentiating twice with respect to $\alpha$ gives
\begin{equation}
\frac{d^2}{d\alpha^2}\widehat{J}_{\mathrm{DEF}}(\sigma\alpha) = \sigma^2 \left( \frac{\varepsilon_n^2}{2}\psi(\alpha) + 2B_f Q(\alpha) \right),
\end{equation}
where $\psi(\alpha) := 2(2\Phi(\alpha)-1) + 2\alpha(4-\alpha^2)\phi(\alpha)$. 
To establish strict convexity, it suffices to show $\psi(\alpha) \ge 0$ for $\alpha \ge 0$, as $B_f > 0$ and $Q(\alpha) > 0$.
The derivative $\psi'(\alpha)=2\phi(\alpha)(\alpha^2-1)(\alpha^2-6)$ shows that $\psi$ increases on $(0,1)$, decreases on $(1,\sqrt{6})$, and increases again on $(\sqrt{6},\infty)$. 
Thus, the only interior local minimum occurs at $\alpha=\sqrt{6}$, while $\alpha=0$ serves as the boundary minimum. 
Since $\psi(0)=0$ and $\psi(\sqrt{6})>0$, we have $\psi(\alpha)> 0$ for all $\alpha> 0$.
Hence, the second derivative of $\widehat{J}_{\mathrm{DEF}}(\sigma\alpha)$ is always positive for all $\alpha> 0$, which establishes the strict convexity of $\widehat{J}_{\mathrm{DEF}}(\sigma\alpha)$ as a function of $\alpha$.
Since $t=\sigma\alpha$ is an affine change of variables, $\widehat{J}_{\mathrm{DEF}}(t)$ is also strictly convex in $t$.
Moreover, 
$\frac{d}{d\alpha}\widehat{J}_{\mathrm{DEF}}(\sigma\alpha)\big|_{\alpha=0}
=
-2\sigma^2B_f\phi(0)<0$,
whereas
$\frac{d}{d\alpha}\widehat{J}_{\mathrm{DEF}}(\sigma\alpha)\to\infty$
as $\alpha\to\infty$.
Thus, the unique minimizer lies in $(0,\infty)$ and is characterized by the first-order condition
$\frac{d}{d\alpha}\widehat{J}_{\mathrm{DEF}}(\sigma\alpha)=0$,
which yields \eqref{eq:DEF_global_optimal}.
\end{IEEEproof}

\begin{example}[ReLU]
For $f(x)=\max\{x,0\}$, we have $B_f=f(1)^2+f(-1)^2=1$.
Under the DEF proxy objective with the endpoint-value substitution, the optimal normalized radius $\alpha^*$ is obtained by solving the following equation:
\begin{equation}\label{eq:relu_opt_cond}
\frac{\varepsilon_n^2}{2}\Bigl(\alpha\bigl(2\Phi(\alpha)-1\bigr)+\alpha^2\phi(\alpha)\Bigr)
+
\Bigl[\alpha Q(\alpha)-\phi(\alpha)\Bigr]
=
0.
\end{equation}
\end{example}

\begin{example}[LeakyReLU]
Let $\lambda\in[0,1)$. 
Define the LeakyReLU by
\begin{equation}\label{eq:lrelu_def}
f_\lambda(x)=
\begin{cases}
x, & x\ge 0,\\
\lambda x, & x<0.
\end{cases}
\end{equation}
We have $B_f=f_\lambda(1)^2+f_\lambda(-1)^2=1+\lambda^2$. 
Under the DEF proxy objective with the endpoint-value substitution, the optimal normalized radius $\alpha^*$ is obtained by solving the following equation:
\begin{equation}\label{eq:lrelu_opt_cond}
\frac{\varepsilon_n^2}{2}\Bigl(\alpha\bigl(2\Phi(\alpha)-1\bigr)+\alpha^2\phi(\alpha)\Bigr)
+
(1+\lambda^2)\Bigl(\alpha Q(\alpha)-\phi(\alpha)\Bigr)
=
0.
\end{equation}
\end{example}

For both ReLU and LeakyReLU, the optimal radius is computed through the following procedure.

\begin{enumerate}
\item Compute the unit-interval minimax error
$\varepsilon_n=E_n(1)$ by solving the degree-$n$ minimax problem on
$[-1,1]$.

\item Substitute the activation-dependent constant $B_f$ into the optimality condition in Theorem~\ref{thm:homogeneous_convexity}, where $B_f=1$ for ReLU and $B_f=1+\lambda^2$ for LeakyReLU. 
The resulting scalar equation in $\alpha$ is then solved by a bisection method. 

\item Set the optimal design radius as $t^*=\sigma\alpha^*,$ and construct the final minimax polynomial $p_{n,t^*}$ on $[-t^*,t^*]$.
\end{enumerate}
Although the homogeneous case requires a numerical minimax computation to obtain the unit-interval error
$\varepsilon_n = E_n(1)$, its dependence on the design radius follows exactly from the scaling law $E_n(t)=t\varepsilon_n$.
We now turn to non-homogeneous activations, for which no such scaling law is generally available and the minimax-error curve
$E_n(t)$ must be estimated.

\subsection{Numerical Surrogate for Non-Homogeneous Functions}\label{subsec:def_nonhomog}

For activations without degree-one positive homogeneity, such as sigmoid, tanh, and GELU, the scaling relation $E_n(t)=tE_n(1)$ no longer holds, so $E_n(t)$ cannot be recovered from a single reference value and must instead be characterized over the range of candidate radii.
We therefore estimate the error--radius curve $t\mapsto E_n(t)$ numerically by sampling minimax errors over candidate design radii and fitting a smooth surrogate $\widetilde E_n(t)$.

Let $\mathcal{T}:=\{t_1,\ldots,t_K\}$ denote the grid of sampled design radii.
For each $t_j\in\mathcal{T}$, we compute
\begin{equation}
p_{n,t_j}=
\arg\min_{p\in\Pi_n}
\|f-p\|_{\infty,[-t_j,t_j]}
\end{equation}
using Chebfun \cite{Driscoll2014chebfun}, and evaluate
\begin{equation}
E_n(t_j)
=
\|f-p_{n,t_j}\|_{\infty,[-t_j,t_j]}.
\end{equation}
Since $E_n(t)$ varies over the candidate range, we construct the surrogate $\widetilde E_n(t)$ by least-squares polynomial regression in $t$; in our experiments we use a degree-$15$ polynomial fit.
The fitted $\widetilde E_n(t)$ is then substituted for $E_n(t)$ in the proxy objective \eqref{eq:J_def_hat}.
 
Fig.~\ref{fig:def_proxy_compare} compares, for representative non-homogeneous functions, the tractable proxy $\widehat{J}_{\mathrm{DEF}}(t)$---evaluated with this fitted surrogate $\widetilde E_n(t)$---against the numerically evaluated ideal objective $J_{\mathrm{DEF}}(t)$ over the sampled candidate radii.
While $J_{\mathrm{DEF}}(t)$ is obtained by direct numerical evaluation of \eqref{eq:J_def_def}, $\widehat{J}_{\mathrm{DEF}}(t)$ combines the arcsine-residual approximation for the within-interval term, and the fitted surrogate $\widetilde E_n(t)$ for the minimax-error curve.
As shown in Fig.~\ref{fig:def_proxy_compare}, $\widehat{J}_{\mathrm{DEF}}(t)$ almost exactly tracks $J_{\mathrm{DEF}}(t)$ near the optimum and accurately reproduces the minimizing radius $t^*$.
This confirms that the proposed proxy preserves the interval selection behavior of the ideal DEF objective and, in particular, validates the surrogate approximation at the radii relevant for interval selection.
The corresponding optimal radii and MSE values are summarized in Tables~\ref{tab:mse_comparison_Gaussian} and~\ref{tab:mse_comparison_Laplace}.

\begin{figure}[!ht]
        \centering
        \subfloat[]{\includegraphics[width=0.45\textwidth]{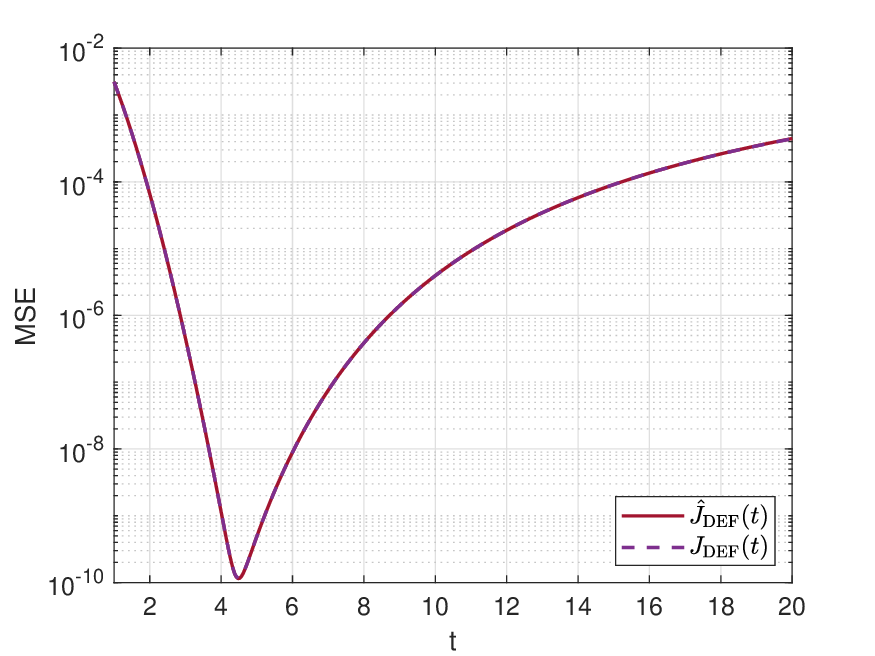}}
        \label{fig:sigmoid_analytic}
        \hfill
        \subfloat[]{\includegraphics[width=0.45\textwidth]{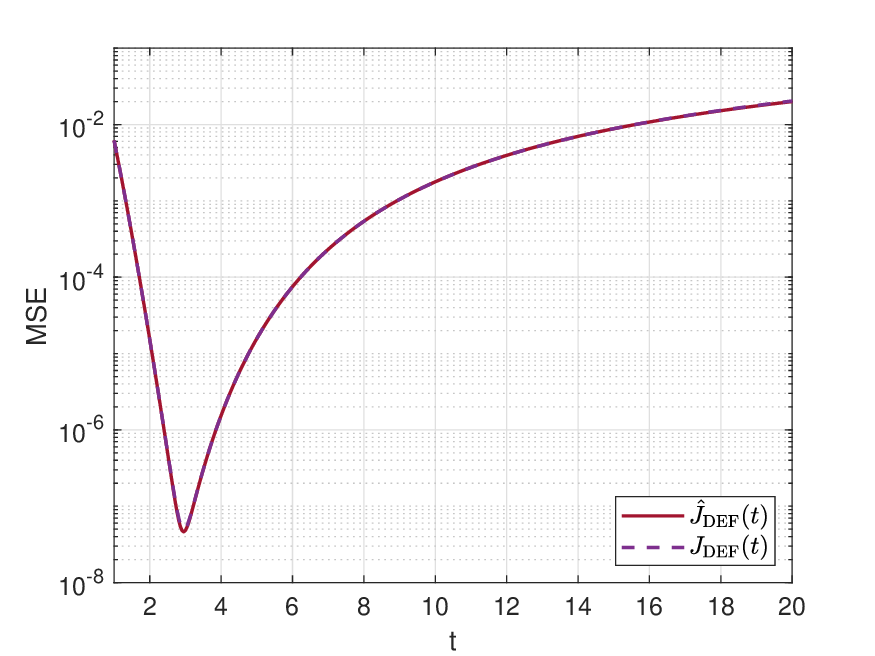}}
        \label{fig:tanh_analytic}
        \hfill
        \subfloat[]{\includegraphics[width=0.45\textwidth]{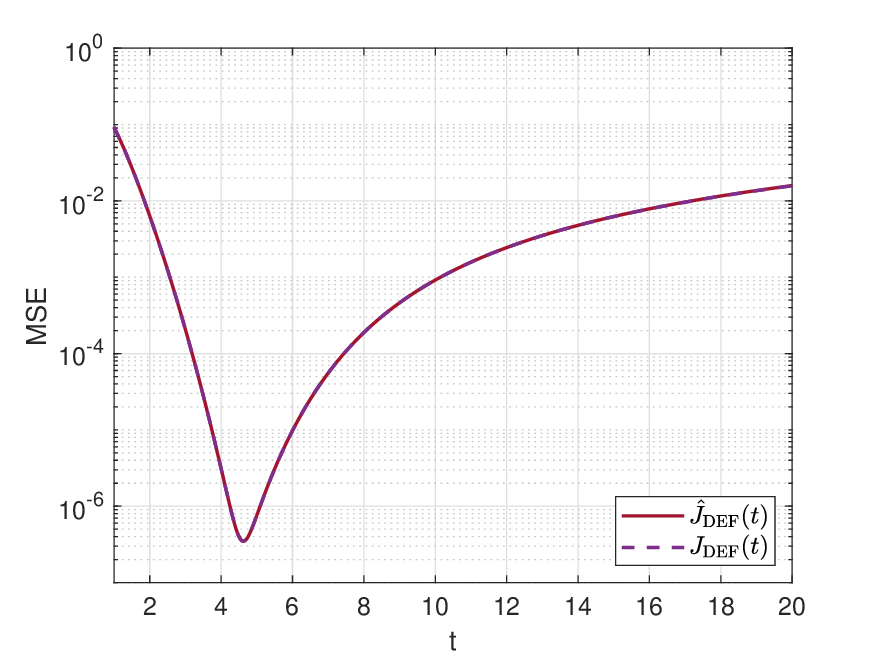}}
        \label{fig:gelu_analytic}
        \caption{Comparison between the numerically evaluated DEF objective $J_{\mathrm{DEF}}(t)$ and the proxy objective $\widehat{J}_{\mathrm{DEF}}(t)$ for three representative nonlinearities: (a) sigmoid, 
        (b) tanh, 
        and (c) GELU.}
    	\vspace{-4mm}
        \label{fig:def_proxy_compare}
    \end{figure}

\section{Interval Optimization via DEP}\label{sec:dep_impl}

In this section, we connect the DEF-based objective to an HE-realizable DEP construction.
The hard-clipping map $D_t(\cdot)$ used in the ideal DEF objective is non-polynomial and therefore cannot be evaluated directly under HE.
As in~\cite{Cheon2022dep}, we replace $D_t(\cdot)$ with a polynomial DEP composition, which provides interval control on a conservative input interval while approximately preserving inputs in a smaller high-probability region.

Let $R_0$ denote a conservative covered-input radius for the pre-activation magnitude.
That is, $R_0$ is chosen so that $\Pr(|X|>R_0)$ is negligible under the pre-activation distribution.
For a given design radius $t \in (0, R_0)$, the DEP composition is then used to map the covered input interval $[-R_0,R_0]$ into the design interval $[-t,t]$.
For each design radius $t$, we choose a base DEP
\begin{equation}
B_t(\cdot)\in\mathcal{D}(\delta_t,\gamma t,t,Lt),
\end{equation}
where $L>1$ is the extension factor, $\delta_t$ controls the identity-preservation error, and $0<\gamma<1$ determines the preservation radius $r=\gamma t$.
Equivalently, this corresponds to the setting in Section~\ref{subsec:def}
\begin{equation}
R_1:=t,\qquad R_2:=Lt,\qquad r:=\gamma t .
\end{equation}
The number of DEP stages is chosen as
\begin{equation}
m_{L}(t)
:=
\max\left\{
0,
\left\lceil
\log_L\left(\frac{R_0}{t}\right)
\right\rceil
\right\}.
\end{equation}
When $t\ge R_0$, no DEP stage is required and $m_{L}(t)=0$.
The $i$-th scaled DEP stage is defined as 
\begin{equation}\label{eq:dep_recur}
B_{t,i}(x)
:=
L^i B_t\left(\frac{x}{L^i}\right),
\qquad
i=0,1,\ldots,m_{L}(t)-1 .
\end{equation}
The resulting DEP composition is
\begin{equation}
g^{(m_{L}(t))}(x)
:=
\bigl(B_{t,0}\circ B_{t,1}\circ\cdots\circ B_{t,m_{L}(t)-1}\bigr)(x).
\end{equation}
By the DEP conditions, the base $B_t(\cdot)$ maps the interval $[-Lt,Lt]$ into the controlled interval $[-t,t]$. 
Consequently, the scaled stage $B_{t,i}$ maps $[-L^{i+1}t,L^{i+1}t]$ into $[-L^i t,L^i t]$. 
Since $m_{L}(t)$ is chosen so that $L^{m_{L}(t)}t\ge R_0$, the composition satisfies
\begin{equation}\label{eq:gt_bound}
g^{(m_{L}(t))}(x)\in[-t,t],
\qquad
\forall x\in[-R_0,R_0].
\end{equation}
This composition serves as an HE-realizable polynomial substitute for the non-polynomial hard-clipping DEF $D_t$. 
On the covered input interval $[-R_0,R_0]$, it maps inputs into the design interval $[-t,t]$.

Using the DEP-realizable map $g^{(m_{L}(t))}$, we define the DEP evaluator as
\begin{equation}\label{eq:q_dep_final}
q_{n,t,\mathrm{DEP}}(x)
:=
p_{n,t}\bigl(g^{(m_{L}(t))}(x)\bigr).
\end{equation}
The corresponding MSE is
\begin{equation}\label{eq:J_dep_final}
J_{\mathrm{DEP}}(t)
:=
\mathbb{E}
\left[
\bigl(f(X)-q_{n,t,\mathrm{DEP}}(X)\bigr)^2
\right].
\end{equation}

To relate this objective to the ideal DEF analysis, recall that
$q_{n,t,\mathrm{DEF}}(x)=p_{n,t}(D_t(x))$.
Then
\begin{align}
f(X)-q_{n,t,\mathrm{DEP}}(X)
&=
\bigl(f(X)-q_{n,t,\mathrm{DEF}}(X)\bigr)
\nonumber\\
&+
\bigl(q_{n,t,\mathrm{DEF}}(X)-q_{n,t,\mathrm{DEP}}(X)\bigr).
\end{align}
We define the DEP implementation error as
\begin{align}\label{eq:J_impl_def}
J_{\mathrm{impl}}(t)
&:=
\mathbb{E}
\left[
\bigl(q_{n,t,\mathrm{DEF}}(X)-q_{n,t,\mathrm{DEP}}(X)\bigr)^2
\right],\\
&\ =
\mathbb{E}
\left[
\bigl(
p_{n,t}(D_t(X))
-
p_{n,t}(g^{(m_{L}(t))}(X))
\bigr)^2
\right].
\end{align}

By the triangle inequality in $L_2$, equivalently Minkowski's inequality~\cite{Stein1970singular},
\begin{equation}\label{eq:J_dep_gap_bound}
J_{\mathrm{DEP}}(t)
\le
\left(
\sqrt{J_{\mathrm{DEF}}(t)}
+
\sqrt{J_{\mathrm{impl}}(t)}
\right)^2.
\end{equation}
Thus, controlling $J_{\mathrm{impl}}(t)$ directly controls the gap between the ideal DEF objective and its DEP realization. 
We next derive an upper bound on $J_{\mathrm{impl}}(t)$ under a covered-input model.

\begin{lemma}[{\cite[Theorem 3]{Cheon2022dep}}]\label{lemma:dep}
Assume that the DEP stages are constructed with extension factor $L>1$ and parameter $\delta_t$, and that the preservation region is $[-r,r]$ with $r=\gamma t$.
Then the composed mapping $g^{(m_{L}(t))}$ satisfies 
\begin{equation}\label{eq:pn_comp_bound}
\left\|
p_{n,t}
-
p_{n,t}\circ g^{(m_{L}(t))}
\right\|_{\infty,[-r,r]}
\le
M_r r^3
\frac{L^2}{L^2-1}
\delta_t ,
\end{equation}
\end{lemma}
where $M_r=\max_{|x|\le r}|p'_{n,t}(x)|$.
\begin{proposition}\label{prop:Jimpl_bound}
Let $R_0$ be a conservative bound on the pre-activation magnitude, and assume that $|X|\le R_0$ under the covered-input model.
Assume further that the DEP composition has enough stages to cover $[-R_0,R_0]$, i.e.,
$L^{m_{L}(t)}t\ge R_0$.
Then the implementation error satisfies
\begin{align}
J_{\mathrm{impl}}(t)
&\le
\left(
M_r r^3
\frac{L^2}{L^2-1}
\delta_t
\right)^2
\Pr(|X|\le r)
\nonumber\\
&\quad+
4\|p_{n,t}\|_{\infty,[-t,t]}^2
\Pr(r<|X|\le R_0).
\label{eq:Jimpl_bound_sup}
\end{align}
\end{proposition}

\begin{IEEEproof}
Under the covered-input condition $|X|\le R_0$, we split the input domain into the preservation region $\{|X|\le r\}$ and the remaining region $\{r<|X|\le R_0\}$.

On $\{|X|\le r\}$, we have $D_t(X)=X$ because $r=\gamma t<t$.
Therefore, Lemma~\ref{lemma:dep} gives
\begin{align}
&\mathbb{E}
\left[
\bigl(
p_{n,t}(D_t(X))
-
p_{n,t}(g^{(m_{L}(t))}(X))
\bigr)^2
\mathbf{1}_{\{|X|\le r\}}
\right]
\nonumber\\
&\quad\le
\left(
M_r r^3
\frac{L^2}{L^2-1}
\delta_t
\right)^2
\Pr(|X|\le r).
\label{eq:Jimpl_core_new}
\end{align}

On $\{r<|X|\le R_0\}$, the hard-clipping map satisfies $D_t(X)\in[-t,t]$, and \eqref{eq:gt_bound} gives $g^{(m_{L}(t))}(X)\in[-t,t]$.
Hence, both arguments of $p_{n,t}(\cdot)$ lie in $[-t,t]$, and therefore
\begin{equation}
\left|
p_{n,t}(D_t(X))
-
p_{n,t}(g^{(m_{L}(t))}(X))
\right|
\le
2\|p_{n,t}\|_{\infty,[-t,t]}.
\end{equation}
It follows that
\begin{align}
&\mathbb{E}
\left[
\bigl(
p_{n,t}(D_t(X))
-
p_{n,t}(g^{(m_{L}(t))}(X))
\bigr)^2
\mathbf{1}_{\{r<|X|\le R_0\}}
\right]
\nonumber\\
&\quad\le
4\|p_{n,t}\|_{\infty,[-t,t]}^2
\Pr(r<|X|\le R_0).
\label{eq:Jimpl_tail_new}
\end{align}
Combining \eqref{eq:Jimpl_core_new} and \eqref{eq:Jimpl_tail_new} proves \eqref{eq:Jimpl_bound_sup}.
\end{IEEEproof}
Define
\begin{align}
U_{\mathrm{impl}}(t)
&:=
\left(
M_r r^3
\frac{L^2}{L^2-1}
\delta_t
\right)^2
\Pr(|X|\le r)
\nonumber\\
&\quad+
4\|p_{n,t}\|_{\infty,[-t,t]}^2
\Pr(r<|X|\le R_0).
\end{align}
Substituting this bound into \eqref{eq:J_dep_gap_bound}, we obtain the final upper bound
\begin{equation}\label{eq:Jdep_gap_final_upper}
J_{\mathrm{DEP}}(t)-J_{\mathrm{DEF}}(t)
\le
2\sqrt{
J_{\mathrm{DEF}}(t)U_{\mathrm{impl}}(t)
}
+
U_{\mathrm{impl}}(t).
\end{equation}

The bound separates the DEP implementation error into two contributions.
The first term in $U_{\mathrm{impl}}(t)$ measures the near-identity error on the preservation region $[-r,r]$ and is controlled by the DEP parameter $\delta_t$, the derivative scale $M_r$, and the preservation radius $r$.
The second term accounts for inputs outside the preservation region but still inside the covered input interval $[-R_0,R_0]$.
Consequently, the DEP-based objective remains close to the ideal DEF objective when the DEP accurately preserves inputs in $[-r,r]$ and the probability mass of $r<|X|\le R_0$ is small.

\section{Numerical Experiments}\label{sec:experiments}

\begin{figure}[!ht]
    \centering
    \subfloat[]{\includegraphics[width=0.45\textwidth]{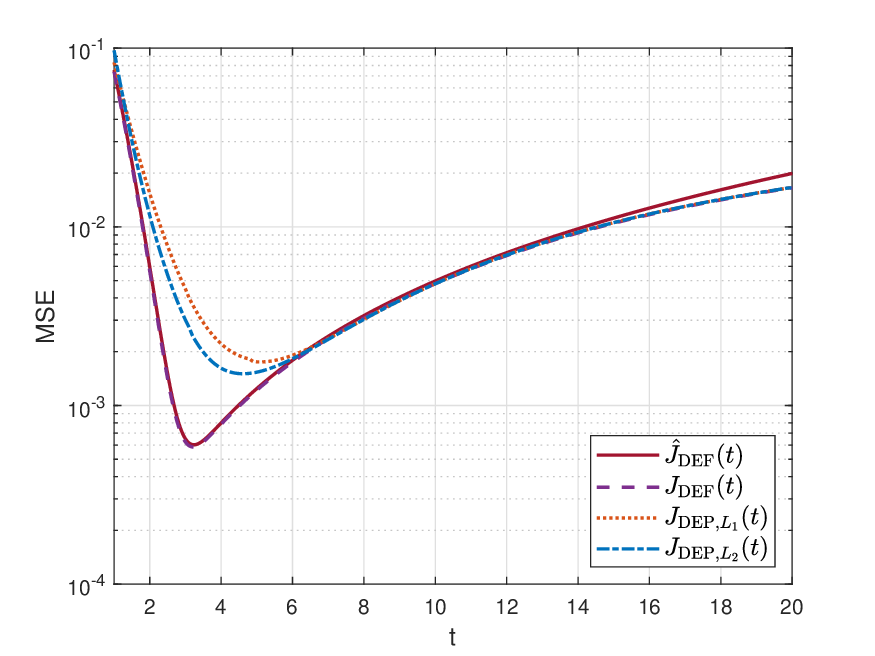}\label{fig:relu_optimal}}
    \hfill
    \subfloat[]{\includegraphics[width=0.45\textwidth]{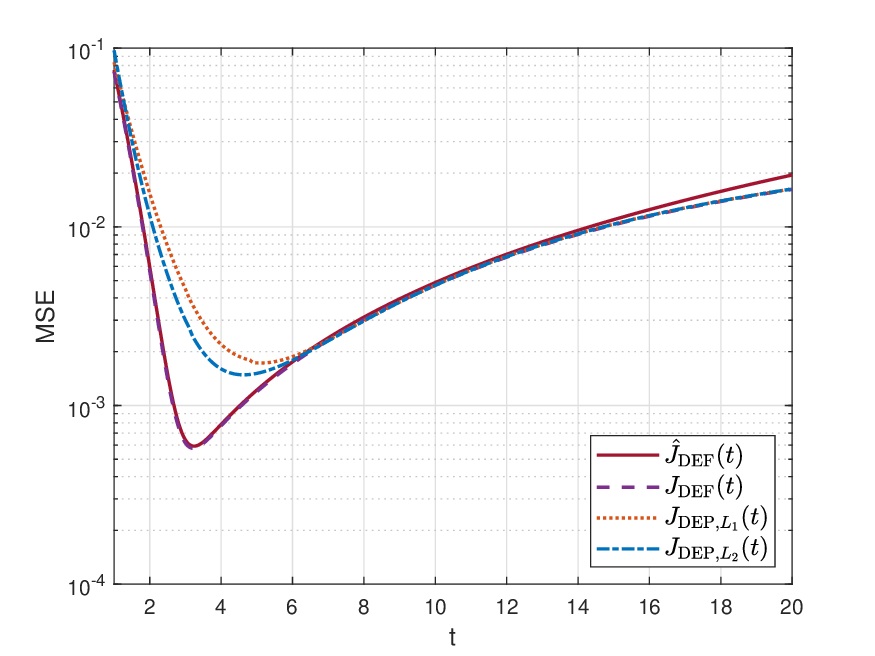}\label{fig:leakyReLU_optimal}}
    \hfill
    \subfloat[]{\includegraphics[width=0.45\textwidth]{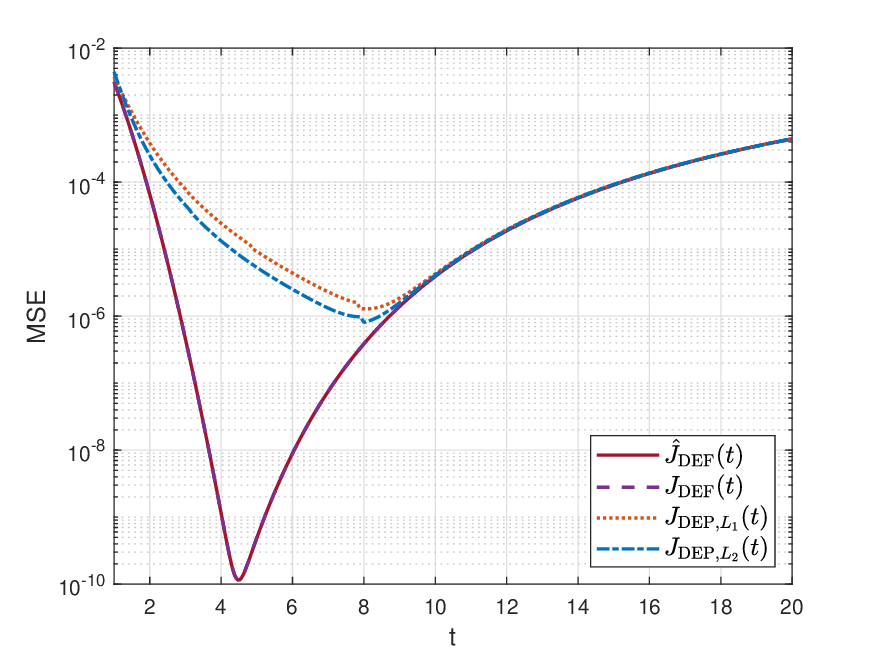}\label{fig:sigmoid_optimal}}
    \caption{Comparison of the analytical proxy $\widehat{J}_{\mathrm{DEF}}(t)$,
the numerical ideal objective $J_{\mathrm{DEF}}(t)$, and the practical
objectives $J_{\mathrm{DEP},L_1}(t)$ and $J_{\mathrm{DEP},L_2}(t)$ for representative nonlinearities:
(a) ReLU, (b) LeakyReLU, and (c) sigmoid.}
    \vspace{-4mm}
    \label{fig:optimal_range1}
\end{figure}

\begin{figure}[!ht]
    \centering
    \subfloat[]{\includegraphics[width=0.45\textwidth]{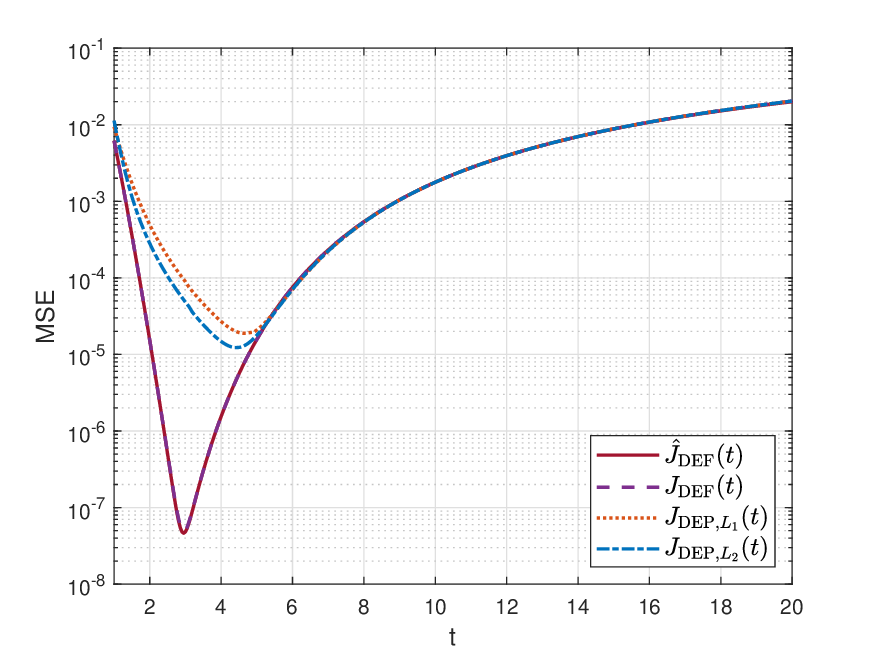}\label{fig:tanh_optimal}}
    \hfill
    \subfloat[]{\includegraphics[width=0.45\textwidth]{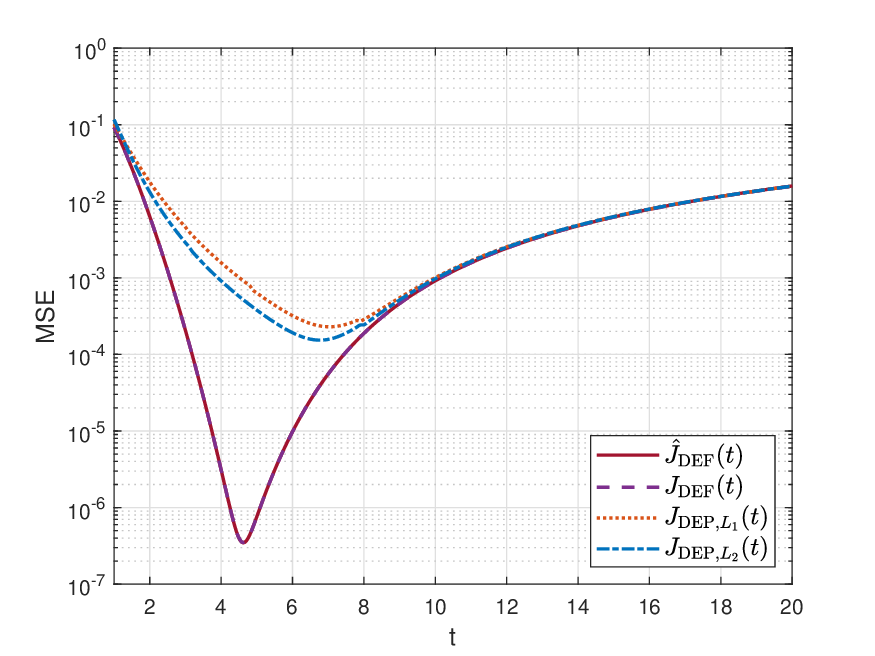}\label{fig:gelu_optimal}}
    \caption{Comparison of the analytical proxy $\widehat{J}_{\mathrm{DEF}}(t)$,
the numerical ideal objective $J_{\mathrm{DEF}}(t)$, and the practical
objectives $J_{\mathrm{DEP},L_1}(t)$ and $J_{\mathrm{DEP},L_2}(t)$ for representative nonlinearities:
(a) tanh, (b) GELU.}
    \vspace{-4mm}
    \label{fig:optimal_range2}
\end{figure}

We evaluate the proposed interval optimization framework under two controlled pre-activation models with unit variance: a Gaussian model $X\sim\mathcal{N}(0,\sigma^2)$ with $\sigma=1$, and a Laplace model $X\sim\mathrm{Laplace}(0,b)$ with $b=1/\sqrt{2}$. 
All numerical minimax approximations are computed using Chebfun~\cite{Driscoll2014chebfun}.
For each candidate radius $t$, we compute the degree-15 minimax polynomial $p_{15,t}$ on $[-t,t]$.
The numerical ideal DEF objective $J_{\mathrm{DEF}}(t)$ is evaluated by directly integrating the ideal-clipping MSE in \eqref{eq:J_def_def}.
The practical DEP objective is evaluated under the covered-input model $|X|\le R_0$. 

For the analytical proxy $\widehat{J}_{\mathrm{DEF}}(t)$, the minimax-error curve $E_n(t)$ is handled according to the activation type.
For ReLU and LeakyReLU with slope $\lambda=0.01$, we use the exact homogeneity-based scaling law $E_n(t)=tE_n(1)$.
For sigmoid, tanh, and GELU, we compute $E_n(t)$ on the sampled grid $\mathcal{T}$ and fit a degree-15 polynomial surrogate $\widetilde{E}_n(t)$ to the sampled values.

The HE-realizable interval control is implemented using the DEP composition introduced in Section~\ref{sec:dep_impl}. 
Because intermediate values cannot be monitored during encrypted execution, the DEP composition must be defined over a covered input region $[-R_0,R_0]$ before evaluation. 
Here, $R_0$ is a coverage parameter that controls the probability of encountering inputs outside the DEP-supported region, whereas the minimax design radius $t$ is optimized separately to minimize the distribution-induced approximation MSE.

We select $R_0$ using an aggregate tail-probability budget. 
Following the encrypted ResNet-20 evaluation setting of \cite{Lee2022low}, the inference involves approximately $N=1.88\times10^9$ scalar activation evaluations on the CIFAR-10 test set and requires
\begin{equation}
    N\cdot\Pr(|X|>R_0)\le \eta,
\end{equation}
with $\eta=10^{-3}$. 
Under the conservative Laplace tail model, this criterion gives $R_0/\sigma\approx19.99$; hence, we use $R_0=20\sigma$ throughout the experiments.

For each candidate design radius $t$, the number of DEP stages is chosen as
\begin{equation}
    m_L(t)
    =
    \max\left\{
    0,
    \left\lceil
    \log_L\left(\frac{R_0}{t}\right)
    \right\rceil
    \right\}.
\end{equation}
To realize HE-compatible interval control, we instantiate the cubic base DEP of~\cite{Cheon2022dep},
\begin{equation}
    B_{\mathrm{cubic}}(x)
    =
    x-\frac{4}{27}x^{3}.
\end{equation}
For a design radius $t$, the scaled cubic DEP is
\begin{equation}
    B_t(x)
    =
    t\,B_{\mathrm{cubic}}(x/t)
    =
    x-\frac{4}{27t^{2}}x^{3}.
\end{equation}

We evaluate two extension factors from the admissible range $1.5<L<1.5\sqrt{3}$, namely $L_1=1.6$ and $L_2=2.5$. A larger $L$ covers $[-R_0,R_0]$ with fewer DEP stages, but it also reduces the admissible preservation region. 
Accordingly, $L_1$ provides a larger preservation region at the cost of more DEP stages, whereas $L_2$ reduces the stage count at the cost of a smaller preservation region. 
This choice exposes the practical trade-off between preservation accuracy and accumulated DEP composition error.

\begin{table*}[!ht]
\centering
\caption{Comparison of the optimal radius $t^*$ and the corresponding objective values for fixed-interval Remez baselines, the analytical proxy $\widehat{J}_{\mathrm{DEF}}$, the numerical ideal objective $J_{\mathrm{DEF}}$, and the practical DEP objectives $J_{\mathrm{DEP},L_1}$ and $J_{\mathrm{DEP},L_2}$ for a Gaussian distribution.}
\label{tab:mse_comparison_Gaussian}
\small
\begin{tabular}{llccccc}
\toprule
\textbf{Function} & \textbf{Metric} & Fixed Remez ($R_0=20$) & \textbf{$\widehat{J}_{\mathrm{DEF}}$} & \textbf{$J_{\mathrm{DEF}}$} & \textbf{$J_{\mathrm{DEP},L_1}$} & \textbf{$J_{\mathrm{DEP},L_2}$} \\
\midrule
\multirow{2}{*}{ReLU} & Radius & $R_0$ & 3.23 & 3.20 & 5.10 & 4.60 \\
 & Value & $1.65 \times 10^{-2}$ & $6.02 \times 10^{-4}$ & $5.86 \times 10^{-4}$ & $1.75 \times 10^{-3}$ & $1.51 \times 10^{-3}$ \\
\midrule
\multirow{2}{*}{LeakyReLU} & Radius & $R_0$ & 3.24 & 3.20 & 5.10 & 4.60 \\
 & Value & $1.62 \times 10^{-2}$ & $5.92 \times 10^{-4}$ & $5.76 \times 10^{-4}$ & $1.73 \times 10^{-3}$ & $1.49 \times 10^{-3}$ \\
\midrule
\multirow{2}{*}{Sigmoid} & Radius & $R_0$ & 4.49 & 4.50 & 8.10 & 8.00 \\
 & Value & $4.40 \times 10^{-4}$ & $1.14 \times 10^{-10}$ & $1.15 \times 10^{-10}$ & $1.29 \times 10^{-6}$ & $8.10 \times 10^{-7}$ \\
\midrule
\multirow{2}{*}{Tanh} & Radius  & $R_0$ & 2.95 & 3.00 & 4.70 & 4.40 \\
 & Value & $2.04 \times 10^{-2}$ & $4.62 \times 10^{-8}$ & $5.12 \times 10^{-8}$ & $1.89 \times 10^{-5}$ & $1.23 \times 10^{-5}$ \\
\midrule
\multirow{2}{*}{GELU} & Radius  & $R_0$ & 4.62 & 4.60 & 7.00 & 6.80 \\
 & Value & $1.57 \times 10^{-2}$ & $3.46 \times 10^{-7}$ & $3.49 \times 10^{-7}$ & $2.30 \times 10^{-4}$ & $1.54 \times 10^{-4}$ \\
\bottomrule
\end{tabular}
\end{table*}

\begin{table*}
\centering
\caption{Comparison of the optimal radius $t^*$ and the corresponding objective values for fixed-interval Remez baselines, the analytical proxy $\widehat{J}_{\mathrm{DEF}}$, the numerical ideal objective $J_{\mathrm{DEF}}$, and the practical DEP objectives $J_{\mathrm{DEP},L_1}$ and $J_{\mathrm{DEP},L_2}$ for a Laplace distribution.}
\label{tab:mse_comparison_Laplace}
\small
\begin{tabular}{llccccc}
\toprule
\textbf{Function} & \textbf{Metric} & Fixed Remez ($R_0=20$) & \textbf{$\widehat{J}_{\mathrm{DEF}}$} & \textbf{$J_{\mathrm{DEF}}$} & \textbf{$J_{\mathrm{DEP},L_1}$} & \textbf{$J_{\mathrm{DEP},L_2}$} \\
\midrule
\multirow{2}{*}{ReLU} & Radius & $R_0$ & 5.11 & 5.10 & 6.50 & 6.10 \\
 & Value & $1.57 \times 10^{-2}$ & $1.66 \times 10^{-3}$ & $1.58 \times 10^{-3}$ & $3.01 \times 10^{-3}$ & $2.60 \times 10^{-3}$ \\
\midrule
\multirow{2}{*}{LeakyReLU} & Radius & $R_0$ & 5.13 & 5.10 & 6.60 & 6.10 \\
 & Value & $1.54 \times 10^{-2}$ & $1.64 \times 10^{-3}$ & $1.56 \times 10^{-3}$ & $2.97 \times 10^{-3}$ & $2.57 \times 10^{-3}$ \\
\midrule
\multirow{2}{*}{Sigmoid} & Radius & $R_0$ & 5.47 & 5.50 & 8.10 & 8.00 \\
 & Value & $3.55 \times 10^{-4}$ & $4.02 \times 10^{-9}$ & $4.04 \times 10^{-9}$ & $1.41 \times 10^{-6}$ & $8.92 \times 10^{-7}$ \\
\midrule
\multirow{2}{*}{Tanh} & Radius & $R_0$ & 3.16 & 3.20 & 4.50 & 4.30 \\
 & Value & $1.80 \times 10^{-2}$ & $1.39 \times 10^{-7}$ & $1.60 \times 10^{-7}$ & $1.57 \times 10^{-5}$ & $1.00 \times 10^{-5}$ \\
\midrule
\multirow{2}{*}{GELU} & Radius & $R_0$ & 6.69 & 6.70 & 8.00 & 8.00 \\
 & Value & $1.83 \times 10^{-2}$ & $7.36 \times 10^{-5}$ & $7.52 \times 10^{-5}$ & $6.34 \times 10^{-4}$ & $4.51 \times 10^{-4}$ \\
\bottomrule
\end{tabular}
\end{table*}

Figs.~\ref{fig:optimal_range1} and~\ref{fig:optimal_range2}, together with Table~\ref{tab:mse_comparison_Gaussian}, summarize the optimal design radii and the corresponding objective values under the Gaussian input model.
The corresponding results under the Laplace input model are reported in Table~\ref{tab:mse_comparison_Laplace}.
For each function, we compare the analytical proxy $\widehat{J}_{\mathrm{DEF}}$, the numerically evaluated ideal DEF objective $J_{\mathrm{DEF}}$, and the practical DEP objectives $J_{\mathrm{DEP},L_1}$ and $J_{\mathrm{DEP},L_2}$.

As a fixed-interval baseline, we use a degree-15 minimax polynomial designed on the conservative interval $[-R_0,R_0]=[-20\sigma,20\sigma]$.
This baseline uses the same degree as the final minimax approximation polynomial $p_{15,t}$ used in the proposed method, but it does not include the additional DEP stages.
Following the cubic DEP construction of~\cite{Cheon2022dep}, each DEP stage requires two non-scalar multiplications.
Therefore, the DEP composition adds approximately $2m_{L}(t)$ non-scalar multiplications and $2m_{L}(t)$ multiplicative levels, where $m_{L}(t)=\lceil\log_L(R_0/t)\rceil$.
Thus, the fixed-interval Remez baseline should be interpreted as a conservative approximation-error reference with the same final minimax degree, rather than as a strict HE-cost-matched baseline.
The practical objectives $J_{\mathrm{DEP},L_1}$ and $J_{\mathrm{DEP},L_2}$ capture the MSE after including the DEP stages.

The primary role of $\widehat{J}_{\mathrm{DEF}}$ is to provide an analytically tractable surrogate for the ideal interval-selection objective $J_{\mathrm{DEF}}$.
Tables~\ref{tab:mse_comparison_Gaussian} and~\ref{tab:mse_comparison_Laplace} show that this surrogate accurately preserves the minimizer structure of the numerical ideal objective.
Across all tested nonlinearities, $\widehat{J}_{\mathrm{DEF}}$ yields minimizing radii that closely match those obtained from the numerically evaluated ideal objective $J_{\mathrm{DEF}}$.
Thus, although $\widehat{J}_{\mathrm{DEF}}$ is not constructed to reproduce the absolute MSE pointwise, it reliably preserves the interval-selection behavior of the ideal DEF objective.

The gap between $J_{\mathrm{DEF}}$ and the practical DEP objectives $J_{\mathrm{DEP},L_1}$ and $J_{\mathrm{DEP},L_2}$ reflects the implementation error introduced by replacing the ideal hard-clipping DEF with a polynomial DEP composition.
In the Gaussian and Laplace settings considered here, $L=2.5$ achieves a smaller practical MSE than $L=1.6$.
Although a smaller extension factor is closer to the ideal clipping behavior at the level of each individual DEP stage, it also requires more stages to cover the same input interval $[-R_0,R_0]$.
These results suggest that the accumulated composition error dominates the per-stage approximation benefit in the tested settings.

Compared with the conservative fixed-interval Remez baseline on $[-20\sigma,20\sigma]$, the optimized DEP-based evaluators achieve substantially lower MSE.
The improvement is especially pronounced for sigmoid, tanh, and GELU, where the achieved MSE is lower by several orders of magnitude.
Table~\ref{tab:mse_comparison_Laplace} shows analogous behavior under the Laplace model.
Compared with the Gaussian case, the heavier tails shift the optimal design radii to larger values, while the analytical proxy still tracks the ideal DEF minimizer closely.

\section{Conclusion} \label{sec:conclusion}

We studied interval optimization for degree-constrained minimax polynomial approximation under homomorphic encryption. 
By treating the approximation interval as an optimization variable, we formulated a distribution-aware objective that explicitly balances the trade-off between within-interval minimax error and outside-interval clipping error. 
We then derived an analytically tractable DEF-based proxy objective and connected it to HE-realizable DEP constructions through an implementation-error decomposition.
The numerical results on representative non-polynomial activation functions show that the minimizer of the analytical proxy closely matches that of the ideal DEF objective across the tested functions.
Moreover, the practical DEP-based evaluators achieve substantially lower MSE than the conservative fixed-interval Remez baseline, with particularly large gains for sigmoid, tanh, and GELU.

Integration of the proposed framework into end-to-end HE-based neural network inference, including the interplay between layer-wise pre-activation distributions, cumulative approximation errors across multiple non-polynomial primitives, and actual HE evaluation costs under specific cryptographic library implementations, is an important direction for future investigation that builds upon the analytical foundation established here.

\bibliographystyle{IEEEtran}
\bibliography{abrv,mybib}
\end{document}